\documentclass[reqno,11pt]{amsart}
\usepackage{cases}
\usepackage{mathrsfs}
\usepackage[T1]{fontenc}
\usepackage{mathrsfs}
\usepackage{amsmath,latexsym,amssymb,amsfonts,amsbsy, amsthm}
\usepackage[usenames]{color}
\usepackage{xspace,colortbl}
\usepackage{epsfig}
\usepackage{amsmath,amsfonts,amsthm,amssymb,amscd}
\input amssym.def
\input amssym.tex

\newcommand{\LC}{\left(}
\newcommand{\RC}{\right)}

\newtheorem{lemma}{Lemma}[section]

\newtheorem{theorem}{Theorem}[section]

\newcommand{\beq}{\begin{equation}}
\newcommand{\eeq}{\end{equation}}
\newcommand{\ben}{\begin{eqnarray}}
\newcommand{\een}{\end{eqnarray}}
\newcommand{\beno}{\begin{eqnarray*}}
\newcommand{\eeno}{\end{eqnarray*}}

\newdimen\eqjot \eqjot = 1\jot
\def\openupeq{\openup \the\eqjot}

\begin{document}

\title [Short-Pulse hierarchy and Sine-Gordon hierarchy]
{Liouville correspondence between the Short-Pulse Hierarchy and the Sine-Gordon Hierarchy}

\author{Jing Kang}
\address{Jing Kang\newline
Center for Nonlinear Studies and School of Mathematics, Northwest University, Xi'an 710069, P. R. China}
\email{jingkang@nwu.edu.cn}
\author{Xiaochuan Liu}
\address{Xiaochuan Liu\newline
Center for Nonlinear Studies and School of Mathematics, Northwest University, Xi'an 710069, P. R. China}
\email{liuxc@nwu.edu.cn}

\begin{abstract}
This paper considers the whole hierarchy of bi-Hamiltonian integrable equations associated to each of the Short-Pulse (SP) equation and the Sine-Gordon (SG) equation. We prove that the transformation that relates the SP equation with the SG equation also serves to establish the correspondence between their flows and Hamiltonian conservation laws in respective hierarchy.
\end{abstract}

\maketitle \numberwithin{equation}{section}




\section{Introduction}

In this paper, we study the correspondence between the Short-Pulse (SP) integrable hierarchy, which is initiated with the SP equation \cite{sw}
\begin{equation}\label{sp}
u_{xt}=u+\frac{1}{6}(u^3)_{xx},
\end{equation}
and the classical Sine-Gordon (SG) integrable hierarchy, which is associated with the SG equation \cite{as, olv0}
\begin{equation}\label{sg}
Q_{y\tau}=\sin Q
\end{equation}
and contains the potential modified KdV (pmKdV) equation
\begin{equation*}
Q_{\tau}=Q_{yyy}+\frac{1}{2}Q_y^3.
\end{equation*}

Recently, much attention devoted to the SP equation \eqref{sp} has been paid in the theory of nonlinear waves because of its significant backgrounds. Physically, equation \eqref{sp} was proposed as a model equation in an appropriate dimensionless form to describe the evolution of ultra-short pulses in nonlinear medium of the silica optical fibre, where $u=u(t, x)$ represents the magnitude of the electric field \cite{sw}. Due to the numerical analysis presented in \cite{cjsw}, the SP equation \eqref{sp} has been recognized as a better alternative of the cubic nonlinear Schr\"{o}dinger equation to approximate the solution of Maxwell equation in the case when the pulse length shortens. Geometrically, equation \eqref{sp} appeared as one of Rabelo's equations associated with pseudo-spherical surfaces \cite{brt, rab}. In addition, the integrability of the SP equation has also been investigated from various points of view. It possesses a zero-curvature representation with a parameter in the context of differential geometry \cite{brt}, admits a Lax pair with the linear spectral problem of the Wadati-Konno-Ichikawa (WKI) type \cite{ss1} and supports the bi-Hamiltonian structure \cite{bru1, bru2}.  

More interestingly, the novel connection which relates the SP equation \eqref{sp} to the classical SG equation \eqref{sg} was found and the associated transformation was used to derive exact soliton solutions of the SP equation from the well-known ones of the SG equation. Some of these resulting solitons of the SP equation are regarded as suitable solutions representing the propagation of ultra-short pulses. Specifically, Sakovich and Sakovich constructed the chain of transformations between \eqref{sp} and \eqref{sg} \cite{ss1, ss3} and obtained the loop-soliton solutions as well as smooth-soliton solutions of the SP equation \cite{ss2}. Matsuno proposed the Hodograph-type transformation between \eqref{sp} and \eqref{sg} to derive the multi-soliton solutions including multi-loop and multi-breather ones, and periodic soliton solutions of the SP equation \cite{mats1, mats2}. Many other results were established with respect to the soliton solutions of the SP equation, see \cite{kbk} and \cite{park}, etc.

Since the SP equation \eqref{sp} and the SG equation \eqref{sg} both admit the bi-Hamiltonian feature, each of them is associated with the corresponding infinite bi-Hamiltonian integrable hierarchy due to the Magri's theory \cite{mag, olv}. In view of the correspondence between the SP equation and the SG equation, it is anticipated that the respective hierarchy should be related in a certain manner. More precisely, can the transformation that converts the SP equation into the SG equation also connects the other flows in the two integrable hierarchies? Is it possible to construct a relationship between the Hamiltonian conservation laws of the SP hierarchy and the SG hierarchy? 

Recently, similar investigation about the correspondence between two integrable hierarchies has been considered conprehensively. For example, the correspondence between the Camassa-Holm (CH) integrable hierarchy and the KdV integrable hierarchy was established by the Liouville transformation in \cite{len0} and \cite{mc}; see also \cite{bss98, bss}. More explicitly, through the Liouville transformation connecting the respective iso-spectral problem, the positive and negative flows of the CH hierarchy are generated by the negative and positive flows of the KdV hierarchy respectively. The correspondence between the Hamiltonian conservation laws of the CH hierarchy and the KdV hierarchy is also derived in \cite{len0}. In a recent paper \cite{kloq}, the authors with collaborators established the Liouville correspondence between the integrable modified Camassa-Holm (mCH) hierarchy and the integrable modified KdV (mKdV) hierarchy (both integrable hierarchies in the negative direction begin with the corresponding Casimir flows). The integrable mCH hierarchy is initiated with the following nonlinear evolution equation
\begin{equation}\label{mch}
m_t+\Big( \LC u^2-u_x^2\RC \,m\Big)_x=0, \qquad m=u-u_{xx},
\end{equation}
known as the mCH equation, that has been studied extensively in recent years (see \cite{bgr, fuc, gloq, llq, or} and references therein).  It is worth noting that the SP equation \eqref{sp} is regarded as the scaling limit equation of the mCH equation \eqref{mch} with the first-order term $u_x$ \cite{gloq}.

The main topic of the present paper is to investigate the correspondence between the integrable SP and SG hierarchies. The corresponding Liouville transformation that takes the following form
\begin{equation}\label{liou}
\cos Q(\tau, y)=\frac{1}{\sqrt{1+u_x^2(t, x)}}, \quad y=\int^x \sqrt{1+u_x^2}, \quad \tau=t,
\end{equation}
not only converts the SP equation \eqref{sp} into the SG equation \eqref{sg}, but also relates the iso-spectral problems of the SP hierarchy and the SG hierarchy. In fact, the transformation \eqref{liou} is equivalent to the transformations given in \cite{mats1, ss1}. Based on the Liouville transformation \eqref{liou}, we are able to construct certain nontrivial relation between the respective recursion operators. Adhering to the particular structures of the corresponding flows, we establish a one-to-one Liouville correspondence between the two integrable hierarchy. Moreover, we study the relationship between the Hamiltonian conservation laws for the SP hierarchy and those for the SG hierarchy through the transformation \eqref{liou}.  The associated conservation laws play a crucial role in the investigation of the qualitative properties such as well-posedness, wave-breaking, etc. for the SP equation \eqref{sp}; see for example \cite{lps} and \cite{ps}. Therefore, the induced relationship between the Hamiltonian conservation laws for the SP equation and the known ones of the SG equation turns out to be of value for studying the SP equation.

The remainder of this paper is organized as follows. In Section 2, we recall some known results on integrability of the mKdV equation, the SP equation, the SG equation and their corresponding hierarchies. The main results in this paper are also presented. In Section 3, we first present the Liouville transformation relating  the iso-spectral problems of the SP hierarchy and the SG hierarchy in Section 3.1. Next in Section 3.2, based on the particular structures of flows in the two hierarchies combined with the relationship between the respective recursion operators, we exploit the Liouville transformation to establish the one-to-one correspondence between the flows in the SP and the SG hierarchies. Section 4 deals with the hierarchy of the Hamiltonian conservation laws of the SP equation and the SG equation. It is proved that the Liouville transformation establishes the correspondence between the series of Hamiltonian conservation laws of the SP equation and the SG equation.

\section{Preliminaries and main results}

We begin with the mKdV equation
\begin{equation}\label{mkdv}
U_{\tau}=U_{yyy}+\frac{3}{2}\,U^2\,U_y,
\end{equation}
which can be written in the bi-Hamiltonian form
\begin{equation*}
U_{\tau}=\mathcal{L}_2\,\frac{\delta E_{1}}{\delta U}=\mathcal{L}_1\,\frac{\delta E_2}{\delta U},
\end{equation*}
where
\begin{equation}\label{haop-mkdv}
\mathcal{L}_1=\partial_y \quad \mathrm{and} \quad \mathcal{L}_2=\partial_y^3+\partial_yU\partial_y^{-1}U\partial_y 
\end{equation}
are the compatible Hamiltonian operators,  while the corresponding Hamiltonian functionals are given by  
\begin{equation*}
E_1(U)=\int \frac{1}{2}\,U^2 \mathrm{d}y \quad \mathrm{and} \quad
E_2(U)=\int \LC -\frac{1}{2}U_y^2+\frac{1}{8}U^4\RC \mathrm{d}y.
\end{equation*}

According to Magri's Theorem \cite{mag,olv}, for an integrable bi-Hamiltonian equation with two compatible Hamiltonian operators $\mathcal{L}_1$ and $\mathcal{L}_2$, we are able to recursively construct an infinite hierarchy 
\begin{equation}\label{hie-mkdv}
U_{\tau}=G_n[U]=\mathcal{L}_2\,\frac{\delta E_{n-1}}{\delta U}=\mathcal{L}_1\,\frac{\delta E_n}{\delta U},\qquad n\in \mathbb{Z},
\end{equation}
of higher-order commuting bi-Hamiltonian integrable systems both in positive and negative directions, based on  the higher-order Hamiltonian functionals $E_{n}$, $n\in \mathbb{Z}$, common to all members of the hierarchy. In the mKdV setting, the members in the hierarchy \eqref{hie-mkdv} are obtained by applying successively the recursion operator $\mathcal{E}=\mathcal{L}_2\,\mathcal{L}_1^{-1}$ to the seed symmetry $G_1[U]=U_y$. Clearly, the mKdV equation \eqref{mkdv} takes the form $U_{\tau}=\mathcal{E}U_y$ and is exactly the second member of the positive flows in this hierarchy. However, in view of the bi-Hamiltonian formulation of the seed equation:
\begin{equation}\label{mkdv-1}
U_{\tau}=G_1[U]=\mathcal{E}^0U_y=\mathcal{L}_2\,\frac{\delta E_{0}}{\delta U}=\mathcal{L}_1\,\frac{\delta E_1}{\delta U}=U_y,
\end{equation}
different choice of the Hamiltonian functional $E_0(U)$ results in different negative flows in the mKdV hierarchy \eqref{hie-mkdv}.

If we take 
\begin{equation*}
E_0(U)=\int \LC \cos\LC\partial_y^{-1}U\RC-1\RC \mathrm{d}y \quad \mathrm{with} \quad \frac{\delta E_0}{\delta U}=\partial_y^{-1}\sin\LC\partial_y^{-1}U\RC-\int \sin\LC\partial_y^{-1}U\RC \mathrm{d}y,
\end{equation*}
then using the sufficiently fast decay condition of $U$ as $|y|\to +\infty$, we have
\begin{equation*}
\begin{aligned}
\mathcal{L}_2\,\frac{\delta E_0}{\delta U} &=\LC \partial_y^2+\partial_yU\partial_y^{-1}U\RC \sin\LC\partial_y^{-1}U\RC\\
&=\partial_y\LC \cos\LC\partial_y^{-1}U\RC\,U\RC-\partial_y\LC U\int_{-\infty}^y\partial_{\xi}\cos\LC\partial_{\xi}^{-1}U\RC \mathrm{d}\xi\RC=U_y.
\end{aligned}
\end{equation*}
Hence, in this case, the negative flows of mKdV hierarchy \eqref{hie-mkdv} begin with
\begin{equation*}
U_{\tau}=G_0[U]=\mathcal{L}_1\,\frac{\delta E_0}{\delta U}=\sin\LC\partial_y^{-1}U\RC,
\end{equation*}
which is just the sine-Gordon (SG) equation
\begin{equation}\label{sg1}
Q_{y\tau}=\sin Q
\end{equation}
for the potential function $Q_y=U$.

Therefore, in such case, the negative flows of the mKdV hierarchy \eqref{hie-mkdv} have the form   
\begin{equation}\label{mkdv-n-1}
U_\tau=G_{-(n-1)}[U]=\LC \mathcal{L}_1\,\mathcal{L}_2^{-1}\RC^{n-1}\sin\LC\partial_y^{-1}U\RC, \quad n=1, 2, \ldots.
\end{equation}
At the $n$th stage, the associated potential function $Q=\partial_y^{-1}U$ satisfies 
\begin{equation*}
Q_\tau=\LC \bar{\mathcal{J}}\bar{\mathcal{K}}^{-1}\RC^{n-1}\bar{\mathcal{J}}\sin Q,
\end{equation*}
where 
\begin{equation}\label{haop-sg}
\bar{\mathcal{K}}=\partial_y+Q_y\partial_y^{-1}Q_y \quad \mathrm{and} \quad \bar{\mathcal{J}}=\partial_y^{-1}
\end{equation}
are the compatible Hamiltonian operators admitted by the SG equation \eqref{sg1}. In view of the Hamiltonian pair \eqref{haop-sg}, the SG equation \eqref{sg1} also admits a hierarchy consisting of an infinite number of integrable bi-Hamiltonian equations in both the positive and negative directions:
\begin{equation}\label{hie-sg}
Q_{\tau}=\bar{K}_n[Q]=\bar{\mathcal{K}}\,\frac{\delta \bar{\mathcal{H}}_{n-1}}{\delta Q}=\bar{\mathcal{J}}\,\frac{\delta \bar{\mathcal{H}}_n}{\delta Q},\qquad n\in \mathbb{Z}.
\end{equation}
These integrable flows in \eqref{hie-sg} could be obtained by applying successively the recursion operator $\bar{\mathcal{R}}=\bar{\mathcal{K}}\,\bar{\mathcal{J}}^{-1}$ to the corresponding seed symmetry, whose equation takes the following form:
\begin{equation*}
Q_{\tau}=\bar{K}_{1}[Q]=\bar{\mathcal{K}}\,\frac{\delta \bar{\mathcal{H}}_0}{\delta Q}=\bar{\mathcal{J}}\,\frac{\delta \bar{\mathcal{H}}_1}{\delta Q}=Q_y
\end{equation*}
with 
\begin{equation}\label{sg-h0h1} 
\bar{\mathcal{H}}_{0}(Q)=\int \LC -\cos Q+1\RC\mathrm{d}y \quad \mathrm{and} \quad \bar{\mathcal{H}}_{1}(Q)=-\frac{1}{2}\int Q_y^2\,\mathrm{d}y.
\end{equation} 
Observe that the SG equation \eqref{sg1} could be expressed exactly by
\begin{equation*}
Q_{\tau}=\bar{K}_{0}[Q]=\bar{\mathcal{R}}^{-1}Q_y=\bar{\mathcal{K}}\,\frac{\delta \bar{\mathcal{H}}_{-1}}{\delta Q}=\bar{\mathcal{J}}\frac{\delta \bar{\mathcal{H}}_0}{\delta Q}=\partial_y^{-1}\,\sin Q
\end{equation*}
with the associated Hamiltonian functional 
\begin{equation*}
\bar{\mathcal{H}}_{-1}(Q)=-\frac{1}{2}\int \cos Q \cdot \big( \partial_y^{-1}\sin Q \big)^2 \mathrm{d}y.
\end{equation*} 
Therefore, the SG equation \eqref{sg1} in its associated hierarchy \eqref{hie-sg} can be viewed as the first member in the negative direction. Furthermore, in the positive direction, the corresponding flows are
\begin{equation*}
Q_{\tau}=\bar{K}_{n}[Q]=\bar{\mathcal{R}}^{n-1}Q_y, \qquad n=1, 2, \ldots,
\end{equation*}
which includes the potential mKdV (pmKdV) equation 
\begin{equation}\label{pmkdv} 
Q_{\tau}=\bar{K}_2[Q]=\bar{\mathcal{R}}\,Q_y=Q_{yyy}+\frac{1}{2}Q_y^3
\end{equation} 
as the second member. 

Whereas, with respect to the integrable mKdV hierarchy \eqref{hie-mkdv}, if we take the Hamiltonian functional $E_0(U)$ for its seed symmetry \eqref{mkdv-1} by
\begin{equation*}
E_0(U)=\int U \,\mathrm{d}y \quad \mathrm{with} \quad \frac{\delta E_0}{\delta U}=1,
\end{equation*}
then, based on the Hamiltonian pair \eqref{haop-mkdv}, it is direct to check that the negative flows of the mKdV hierarchy \eqref{hie-mkdv} are generated from the Casimir equation
\begin{equation*}
U_{\tau}=G_{-1}[U]=\mathcal{L}_1\,\frac{\delta E_C}{\delta U} \quad \mathrm{with} \quad \mathcal{L}_2\,\frac{\delta E_C}{\delta U}=\mathcal{L}_1\,\frac{\delta E_0}{\delta U}=0.
\end{equation*}
The corresponding negative flows $U_{\tau}=G_{-n}[U]$, $n=1, 2, \ldots$, can be rewritten as 
\begin{equation*}
\LC \mathcal{L}_2\,\mathcal{L}_1^{-1}\RC^n U_{\tau}=\partial_y\LC \partial_y+U\partial_y^{-1}U\RC \LC \mathcal{L}_2\,\mathcal{L}_1^{-1}\RC^{n-1}U_{\tau}=0.
\end{equation*}
Integrating the above equation, we have
\begin{equation}\label{mkdv-n-2}
\LC \partial_y+U\partial_y^{-1}U\RC \LC \mathcal{L}_2\,\mathcal{L}_1^{-1}\RC^{n-1}U_\tau=C_{-n},
\end{equation}
with $C_{-n}$ being the corresponding constants of integration for the flows $U_{\tau}=G_{-n}[U]$, $n=1, 2, \ldots$. It was proved in \cite{kloq} that, in such case, the Liouville transformation that relates the corresponding iso-spectral problems establishes the one-to-one correspondence between the flows in the integrable mCH hierarchy initiated with the mCH equation \eqref{mch} and the mKdV hierarchy with negative flows given by \eqref{mkdv-n-2}. More precisely, for each $n\in  \mathbb{Z}^+$, under the Liouville transformation, the $(n+1)$th equation in the positive mCH hierarchy is mapped into the $n$th equation \eqref{mkdv-n-2} in the negative mKdV hierarchy, and conversely. In the opposite direction, the $(n+1)$th equation in the positive mKdV hierarchy is also related to the $n$th equation in the negative mCH hierarchy. 

However, this brings out a natural question, whether there exists the correspondence in a certain manner which can relate the mKdV hierarchy with negative flows given by \eqref{mkdv-n-1} or more suitably the integrable sine-Gordon hierarchy \eqref{hie-sg} to some other integrable hierarchy? It was proved in \cite{ss1} and \cite{ss3} that the SG equation \eqref{sg1} and the SP equation \eqref{sp} are related by the following chain of coordinate transformations
\begin{equation}\label{equi}
 v(t, x)=\frac{1}{\sqrt{1+u_x^2}},\quad x=w(t, y),\quad v(t,\,x)=w_y(t, y),\quad \tau=t, \quad Q(\tau, y)=\arccos w_y.
\end{equation}
Moreover, a Hodograph-type transformation also reveals the connection between the SP equation and the SG equation \cite{mats1}. In view of these results, it is anticipated to extend such a relationship between the SG equation and the SP equation to their respective integrable hierarchies. In other words, we desire to establish the one-to-one correspondence between the whole flows in the SG hierarchy and the SP hierarchy, as well as the Hamiltonian functionals involved. 

As far as the SP hierarchy is concerned, it is associated with the SP equation \eqref{sp}, which can be rewritten as 
\begin{equation}\label{sp1}
u_t=\partial_x^{-1} u+\frac{1}{2}u^2u_x
\end{equation}
and admits the bi-Hamiltonian formulation of the following form \cite{bru1, bru2}
\begin{equation*}
u_t=\mathcal{K}\,\frac{\delta \mathcal{H}_1}{\delta u}=\mathcal{J}\,\frac{\delta \mathcal{H}_2}{\delta u},
\end{equation*}
where the compatible Hamiltonian operators are given by
\begin{equation}\label{sp-bihaop}
\mathcal{K}=\partial_x^{-1}+u_x\,\partial_x^{-1}\,u_x \qquad \mathrm{and} \qquad
\mathcal{J}=\partial_x
\end{equation}
and the associated  Hamiltonian functionals are
\begin{equation}\label{sp-ha12}
\mathcal{H}_1(u)=\int \frac{1}{2}\,u^2\mathrm{d}x \quad \mathrm{and} \quad
\mathcal{H}_2(u)=\int \LC -\frac{1}{2}(\partial_x^{-1}u)^2+\frac{1}{24}u^4\RC\mathrm{d}x.
\end{equation}
Therefore, the integrable SP hierarchy which consists of an infinite number of higher-order bi-Hamiltonian systems
\begin{equation}\label{hie-sp}
u_t=K_n[u]=\mathcal{K}\,\frac{\delta \mathcal{H}_{n-1}}{\delta u}=\mathcal{J}\,\frac{\delta \mathcal{H}_n}{\delta u},\qquad n \in \mathbb{Z},
\end{equation}
can also be readily constructed by applying successively the recursion operator $\mathcal{R}=\mathcal{K}\,\mathcal{J}^{-1}$ to the seed symmetry $K_1[u]=u_x$. It is obvious that the SP equation \eqref{sp1} in this hierarchy is exactly the second member $u_t=K_2[u]=\mathcal{R} u_x$ of the positive flows. While, in the negative direction, it follows from 
\begin{equation*}
u_t=K_1[u]=\mathcal{K}\,\frac{\delta \mathcal{H}_0}{\delta u}=\mathcal{J}\,\frac{\delta \mathcal{H}_1}{\delta u}=u_x
\end{equation*}
with the associated Hamiltonian functional
\begin{equation}\label{sp-h0}
\mathcal{H}_0(u)=\int \LC -\sqrt{1+u_x^2}+1\RC \mathrm{d}x \quad \mathrm{and} \quad \frac{\delta \mathcal{H}_0}{\delta u}=\LC \frac{u_{x}}{\sqrt{1+u_x^2}}\RC_x,
\end{equation}
that the negative flows in the SP hierarchy \eqref{hie-sp} begin with the following equation 
\begin{equation}\label{wki}
u_t=K_{0}[u]=\mathcal{R}^{-1}u_x=\mathcal{J}\,\frac{\delta \mathcal{H}_0}{\delta u}=\LC \frac{u_{x}}{\sqrt{1+u_x^2}}\RC_{xx},
\end{equation}
which is known as the WKI equation describing the nonlinear transverse oscillations of elastic beams under tension \cite{wki}. Actually, it is implies in \cite{ss1} that the WKI equation \eqref{wki} is connected by the transformation \eqref{equi} with the pmKdV equation \eqref{pmkdv} which belongs to the SG hierarchy  \eqref{hie-sg}. Therefore, these arguments indicate to consider the correspondence between the SP hierarchy \eqref{hie-sp} and the SG hierarchy \eqref{hie-sg}.

In the present paper, with the aim to investigate the correspondence between the SG hierarchy and the SP hierarchy, we start from the perspective of the iso-spectral problems for the SG hierarchy 
\begin{equation}\label{iso-sg}
\mathbf{\Phi}_y
=\begin{pmatrix} \lambda \cos Q \; & \lambda \sin Q \\
                \lambda \sin Q\; & -\lambda \cos Q 
               \end{pmatrix}   \mathbf{\Phi},
\end{equation} and the SP hiearachy
\begin{equation}\label{iso-sp}
\mathbf{\Psi}_x
=\begin{pmatrix}\mu\; &\mu u_x\\
              \mu u_x\; &-\mu
               \end{pmatrix}\mathbf{\Psi},
\end{equation}
where  $\lambda$ and $\mu$ are the respective spectral parameters. And then we show that the Liouville transformation
\begin{equation}\label{liouville}
Q=\arccos\frac{1}{\sqrt{1+u_x^2}},\quad y=\int^x \sqrt{1+u_\xi^2} \,\mathrm{d}\xi
\end{equation}
relating the isospectral problems \eqref{iso-sg} and \eqref{iso-sp} will establish the one-to-one correspondence between the flows in the the SP hiearachy
and the SG hierarchy. Usually, the process of going from one spectral problem to another one by means of a change of variables has been recognized as a form of the classical Liouville transformation, which arises naturally in the context of the so-called WKB approximation; see \cite{mil, folv}.  Note further that transformation  \eqref{liouville}  together with $\tau=t$ are actually equivalent to the chain of transformations \eqref{equi} and the Hodograph-type transformation introduced in \cite{mats1}. The main result on the Liouville correspondence between two hierarchies is the following. 

\begin{theorem}\label{t1.1}
For any integer $n\in \mathbb{Z}$, the $(n+1)$th equation $u_t=K_{n+1}[u]$ in the SP hierarchy \eqref{hie-sp} is mapped into $(1-n)$th equation $Q_{\tau}=\bar{K}_{1-n}[Q]$ in the SG hierarchy \eqref{hie-sg} under the Liouville transformation \eqref{liouville} and $\tau=t$, and conversely.
\end{theorem}

We next focus our attention on the Hamiltonian conservation laws admitted by the SP equation and the SG equation. On the one hand, the compatible bi-Hamiltonian structure  \eqref{hie-sp} produces the recursively constructed infinite sequence of Hamiltonian functionals in both the negative and positive directions:
\begin{equation}\label{conslaws}
\ldots, \mathcal{H}_{-2}, \,\mathcal{H}_{-1},\,\mathcal{H}_0,\,\mathcal{H}_1,\,\mathcal{H}_2,\,\ldots,
\end{equation}
which are all conserved densities of the SP equation \eqref{sp1}. On the other hand, the recursive formula \eqref{hie-sg} gives rise to an infinite sequence of Hamiltonian functionals
\begin{equation}\label{cl-sg}
\ldots, \bar{\mathcal{H}}_{-2}, \,\bar{\mathcal{H}}_{-1},\,\bar{\mathcal{H}}_0,\,\bar{\mathcal{H}}_1,\,\bar{\mathcal{H}}_2,\,\ldots,
\end{equation}
conserved under the SG flow \eqref{sg1} \cite{mag,olv}. We will study the correspondence between the conserved quantities in the hierarchies \eqref{conslaws} and \eqref{cl-sg}, and prove that the Liouville transformation \eqref{liouville} not only links the integrable flows in the respective hierarchies but also relates the corresponding Hamiltonian conservation laws admitted by the two integrable equations.  More precisely, we establish the following theorem, illustrating the preceding claim.

\begin{theorem}\label{t1.2}
For any integer $n$, each Hamiltonian conservation law $\bar{\mathcal{H}}_n(Q)$ of the SG equation in \eqref{cl-sg} yields the Hamiltonian conservation law $\mathcal{H}_{-n}(u)$ of the SP equation in \eqref{conslaws}, under the Liouville transformation \eqref{liouville}, according to the following identity
\begin{equation*}
 \mathcal{H}_{-n}(u)=-\bar{ \mathcal{H}}_{n}(Q), \qquad  n\in \mathbb{Z}.
\end{equation*}
\end{theorem}

\section{The correspondence between the short pulse and sine-Gordon hierarchies}

\subsection{A Liouville transformation between the iso-spectral problems of the SP and SG equations.}

Let us begin with the iso-spectral problems associated to the SP equation \eqref{sp1} and the SG equation \eqref{sg1}. It has been proposed in \cite{ss1} that the zero curvature formulation for the SP equation \eqref{sp1} takes the form
 \begin{equation*}
\mathbf{D}_t\mathbf{M}-\mathbf{D}_x\mathbf{N}+[\mathbf{M},\;\mathbf{N}]=0
\end{equation*}
with
 \begin{equation*}
\mathbf{M}=\begin{pmatrix} 
                      \mu\; &\mu u_x\\
                      \mu u_x\; &-\mu
                    \end{pmatrix} 
               \qquad \mathrm{and} \qquad 
 \mathbf{N}=\begin{pmatrix} 
                      \frac{\mu}{2}u^2+\frac{1}{4\mu}\; & \frac{\mu}{2}u^2u_x-\frac{1}{2}u\\
                      \frac{\mu}{2}u^2u_x+\frac{1}{2}u\; &-\frac{\mu}{2}u^2-\frac{1}{4\mu}
                     \end{pmatrix},
\end{equation*}
which gives rise to the linear iso-spectral problem associated to the SP equation \eqref{sp1}, namely
\begin{equation}\label{isosp}
\mathbf{\Psi}_x =\mathbf{M} \mathbf{\Psi},\qquad  
\mathbf{\Psi}=\begin{pmatrix}
                          \psi_1 \\ 
                          \psi_2
                       \end{pmatrix}.
\end{equation}
On the other hand, the SG equation \eqref{sg1} arises from the compatibility condition $\partial_\tau(\mathbf{\Phi}_y)=\partial_y (\mathbf{\Phi}_\tau)$, where 
\begin{equation}\label{isosg}
\mathbf{\Phi}_y=\begin{pmatrix} 
                             \lambda \cos Q \; &\lambda \sin Q\\
                             \lambda \sin Q\; &-\lambda \cos Q
                           \end{pmatrix} \mathbf{\Phi},\qquad
\mathbf{\Phi}=\begin{pmatrix}
                         \phi_1 \\ 
                         \phi_2
                        \end{pmatrix},
\end{equation}
and
\begin{equation*}
\mathbf{\Phi}_\tau=\begin{pmatrix} 
                                 \frac{1}{4 \lambda} &-\frac{1}{2}Q_\tau\\
                                 \frac{1}{2}Q_\tau &-\frac{1}{4 \lambda}
                               \end{pmatrix}\mathbf{\Phi}.
\end{equation*}

One can verify that the following Liouville transformation
\begin{equation*}
\mathbf{\Phi}=\mathbf{\Psi},\qquad y=\int^x \sqrt{1+u_z^2}\,\mathrm{d}z
\end{equation*}
will convert the linear spectral problem (\ref{isosp}) into the  linear spectral problem (\ref{isosg}), with
\begin{equation*}
Q=\arccos\frac{1}{\sqrt{1+u_x^2}} \qquad \mathrm{and} \qquad \lambda=\mu .
\end{equation*}
This allows us introduce the following coordinate transformations, say
\begin{equation}\label{recixy}
 y=\int^x \sqrt{1+u_z^2(t, z)}\,\mathrm{d}z, \qquad \tau=t
\end{equation}
and
\begin{equation}\label{reciqu}
Q(\tau, y)=\arccos\frac{1}{\sqrt{1+u_x^2(t, x)}}.
\end{equation}
Transformations \eqref{recixy} and \eqref{reciqu} are, in fact, equivalent to the chain of transformations \eqref{equi} and the Hodograph-type transformation introduced in \cite{mats1} and serves to connect the SP equation \eqref{sp1} with the SG equation \eqref{sg1}. Note that the first equation in transformation \eqref{recixy}  has the form of the reciprocal transformation, which exchanges the roles of the dependent and independent variables. In the following subsection, we will investigate how the Liouville transformations \eqref{recixy} and \eqref{reciqu} relating the spectral problems \eqref{isosp} and  \eqref{isosg} affect the underlying correspondence between the flows in the SP and SG hierarchies.

\subsection{The correspondence between the SP and SG hierarchies.}

We now focus our attention on the SP and SG hierarchies. First of all, if we apply the recursion operator 
\begin{equation}\label{sp-re}
\mathcal{R}=\mathcal{K}\,\mathcal{J}^{-1}=\left(\partial_x^{-1}+u_x\,\partial_x^{-1}\,u_x\right)\partial_x^{-1}
\end{equation} 
of the SP equation \eqref{sp1} successively to the translational symmetry with characteristic $K_1[u]=u_x$, we can obtain the positive flows
\begin{equation}\label{spn+1}
u_t=K_{n+1}[u]=\mathcal{R}^n\,u_x,\qquad n=0, 1, \ldots,
\end{equation} 
in the SP hierarchy \eqref{hie-sp}. On the other hand, starting with 
\begin{equation*}
u_t=K_{0}[u]=\partial_x^2\frac{u_x}{\sqrt{1+u_x^2}},
\end{equation*} 
the $(n-1)$th member in the negative flows of the SP hierarchy \eqref{hie-sp} can be written as
\begin{equation}\label{sp-n}
u_t=K_{-(n-1)}[u]=\mathcal{R}^{-(n-1)}\,\partial_x^2\frac{u_x}{\sqrt{1+u_x^2}},\quad n=1, 2,\ldots.
\end{equation}

Similarly, for the SG hierarchy \eqref{hie-sg}, the positive flows take the following form
\begin{equation}\label{sgn+1}
Q_\tau=G_{n+1}[Q]=\bar{\mathcal{R}}^n\,Q_y,\qquad n=0, 1,\ldots,
\end{equation}
where the corresponding recursion operator is 
\begin{equation}\label{sg-re}
\bar{\mathcal{R}}=\bar{\mathcal{K}}\,\bar{\mathcal{J}}^{-1}=\partial_y^2+Q_y\partial_y^{-1}Q_y\partial_y.
\end{equation}
While, in the negative direction, since $\delta\bar{\mathcal{H}}_0/\delta Q=\sin Q$,  the $(n-1)$th negative flow $Q_{\tau}=\bar{K}_{n-1}[Q]$, $n=1, 2, \ldots$, can be expressed by
\begin{equation}\label{sg-n}
\bar{\mathcal{J}}^{-1}\bar{\mathcal{R}}^{(n-1)}\,Q_\tau=\sin Q, \qquad n=1, 2,\ldots.
\end{equation}

Hereafter, as a matter of convenience, for each positive integer $n$, we write the $n$th equation in the positive direction of the SP and SG hierarchies by 
$(SP)_n$ and $(SG)_n$, respectively. While for each non-negative integer $n$, the $n$th negative flow in the SP and SG hierarchies are denoted by $(SG)_{-n}$ and $(SG)_{-n}$, respectively. With these notations, we now restate Theorem \ref{t1.1} and present the explicit description of the correspondence between the two hierarchies.

\begin{theorem}\label{t3.1}
Under the transformations \eqref{recixy} and \eqref{reciqu}, for each $l\in \mathbb{Z}$, the $(SP)_{l+1}$ equation is related to the $(SG)_{1-l}$ equation. More precisely, 

{\bf(i).}  for each integer $n \geq 0$, $u$ is a solution of  the $(SP)_{n+1}$ equation \eqref{spn+1} if and only if $Q$ satisfies $Q_\tau=Q_y$ for $n=0$ or the $(SG)_{-(n-1)}$ equation \eqref{sg-n} for $n \geq 1$;

{\bf(ii).} for each integer $n\geq 1$, $u$ solves the $(SP)_{-(n-1)}$ equation \eqref{sp-n} if and only if $Q$ satisfies the $(SG)_{n+1}$ equation
\eqref{sgn+1}. 
\end{theorem}

The proof of Theorem \ref{t3.1} relies on the following lemma, which sets up the relationship between the recursion operators $\mathcal{R}$ \eqref{sp-re} and $\bar{\mathcal{R}}$ \eqref{sg-re} admitted by SP and SG hierarchies respectively with $u(t, x)$ and $Q(\tau, y)$ connected by the transformations \eqref{recixy} and \eqref{reciqu}.

\begin{lemma}\label{l3.1}
Let $\mathcal{R}$ be the recursion operator \eqref{sp-re} for the SP hierarchy and $\bar{\mathcal{R}}$ be the recursion operator \eqref{sg-re} for the SG hierarchy. Then, for each integer $n\geq 1$,
\begin{equation}\label{reop-l3.1}
\bar{\mathcal{R}}^n=\partial_x^{-1}\LC \mathcal{R}^{-1}\RC^n \partial_x,
\end{equation}
under the transformations \eqref{recixy} and \eqref{reciqu}.
\end{lemma}

\begin{proof}
We prove \eqref{reop-l3.1} by induction in $n$. For the case $n=1$, based on the form of the operator $\mathcal{R}$ \eqref{sp-re}, it suffices to prove the following operator identity
\begin{equation}\label{l3.1-1}
\mathcal{R}\,\partial_x\bar{\mathcal{R}}=\LC \partial_x^{-1}+u_x\partial_x^{-1}u_x\RC \bar{\mathcal{R}}=\partial_x.
\end{equation}
In view of the transformations \eqref{recixy} and \eqref{reciqu}, we have several expressions after direct calculation
\begin{equation}\label{reci-1}
\partial_x=\sqrt{1+u_x^2}\,\partial_y, \qquad \partial_x^{-1}=\partial_y^{-1}\sqrt{1+u_x^2}=\partial_y^{-1}\cos Q
\end{equation} 
and
\begin{equation}\label{reci-2}
\cos Q=\frac{1}{\sqrt{1+u_x^2}}, \qquad \sin Q=\frac{u_x}{\sqrt{1+u_x^2}}.
\end{equation}
Hence, for a test function $\rho\in \mathcal{C}_c^{\infty}(\mathbb{R})$, using \eqref{reci-1}, \eqref{reci-2} and integration by parts, we have
\begin{equation}\label{l3.1-2}
\begin{aligned}
\partial_x^{-1}\bar{\mathcal{R}}\rho&=\partial_y^{-1}\LC \cos Q\cdot \bar{\mathcal{R}}\rho\RC=\partial_y^{-1}\LC \cos Q\cdot \rho_{yy}+\cos Q \cdot Q_y\partial_y^{-1}Q_y\rho_y\RC\\
&=\int^y \cos Q \cdot \rho_{\xi\xi}\, \mathrm{d}\xi + \int^y \cos Q \cdot Q_{\xi}\partial_{\xi}^{-1}Q_{\xi}\rho_{\xi}\, \mathrm{d}\xi\\
&=\cos Q\cdot \rho_y+\int^y \sin Q\cdot Q_{\xi}\rho_{\xi}\, \mathrm{d}\xi+\int^y (\sin Q)_{\xi}\partial_{\xi}^{-1}Q_{\xi}\rho_{\xi}\, \mathrm{d}\xi\\
&=\cos Q\cdot \rho_y+\sin Q\cdot \partial_y^{-1}Q_y\rho_y
\end{aligned}
\end{equation}
and
\begin{equation}\label{l3.1-3}
\begin{aligned}
u_x\partial_x^{-1}u_x\,\bar{\mathcal{R}}\rho & =u_x\partial_y^{-1}\frac{u_x}{\sqrt{1+u_x^2}}\bar{\mathcal{R}}\rho=\tan Q\cdot \int^y \sin Q \cdot \LC \rho_{\xi\xi}+Q_{\xi}\partial_{\xi}^{-1}Q_{\xi}\rho_{\xi}\RC\mathrm{d}\xi\\
&=\tan Q\cdot \LC \int^y \sin Q\cdot \rho_{\xi\xi}\,\mathrm{d}\xi + \int^y \sin Q\cdot Q_{\xi}\partial_{\xi}^{-1}Q_{\xi}\rho_{\xi}\,\mathrm{d}\xi\RC\\
&=\tan Q\cdot \LC \sin Q\cdot \rho_y-\int^y \cos Q\cdot Q_{\xi}\rho_{\xi}\, \mathrm{d}\xi-\int^y (\cos Q)_{\xi}\partial_{\xi}^{-1}Q_{\xi}\rho_{\xi}\, \mathrm{d}\xi\RC\\
&=\tan Q\cdot \LC \sin Q\cdot \rho_y - \cos Q \cdot \int^y Q_{\xi}\rho_{\xi}\,\mathrm{d}\xi\RC.
\end{aligned}
\end{equation}
Combining \eqref{l3.1-2} with \eqref{l3.1-3} gives rise to
\begin{equation*}
\LC \partial_x^{-1}+u_x\partial_x^{-1}u_x\RC \bar{\mathcal{R}}\rho=\cos Q\cdot \rho_y+\frac{\sin^2 Q}{\cos Q}\cdot \rho_y=\frac{1}{\cos Q}\cdot \rho_y=\rho_x
\end{equation*}
and verifies the identity \eqref{l3.1-1}.

Finally, for the general case, we assume  \eqref{reop-l3.1} holds for $n=k$, in other words,
\begin{equation*}
\bar{\mathcal{R}}^k=\partial_x^{-1}\LC \mathcal{R}^{-1}\RC^k \partial_x.
\end{equation*}
Then, for $n=k+1$, the result of $n=1$ readily leads to 
\begin{eqnarray*}
\LC \bar{\mathcal{R}}\RC^{k+1}=\partial_x^{-1}\LC \mathcal{R}^{-1}\RC^k \partial_x \bar{\mathcal{R}}=\partial_x^{-1}\LC \mathcal{R}^{-1}\RC^k \partial_x \partial_x^{-1}\,\mathcal{R}^{-1}\,\partial_x=\partial_x^{-1}\, \LC \mathcal{R}^{-1}\RC^{k+1}\,\partial_x,
\end{eqnarray*}
illustrating that \eqref{reop-l3.1} holds for each $n\geq1$. Therefore, the lemma is proved.
\end{proof}

\begin{proof} [\bf{Proof of Theorem \ref{t3.1}}]

{\bf(i).} We begin with the $(SP)_{n+1}$ equation \eqref{spn+1} for $n\geq 1$.  Note that the equation \eqref{spn+1} can be written as 
\begin{eqnarray}\label{thm3.1-1}
u_t=\left(\partial_x^{-1}+u_x\,\partial_x^{-1}\,u_x\right)\partial_x^{-1}\mathcal{R}^{n-1}\,u_x, \quad n=1, 2, \ldots.
\end{eqnarray}
Suppose that $u=u(t, x)$ is the solution of equation \eqref{thm3.1-1}. We first calculate the $t$-derivative of the new variable $y$ defined in \eqref{recixy} for the solution $u(t, x)$. More precisely, using \eqref{thm3.1-1}, we have
\begin{equation}\label{thm3.1-2}
\begin{aligned}
y_t & =\int^x \frac{u_zu_{zt}}{\sqrt{1+u_z^2}}\,\mathrm{d}z=\int^x \frac{u_z}{\sqrt{1+u_z^2}}\,\partial_z\LC \partial_z^{-1}+u_z\partial_z^{-1}u_z\RC\partial_z^{-1}\mathcal{R}^{n-1}u_z\, \mathrm{d}z\\ 
& =\int^x \frac{u_z}{\sqrt{1+u_z^2}}\,\LC 1+u_z^2+u_{zz}\partial_z^{-1}u_z\RC \partial_z^{-1}\mathcal{R}^{n-1}u_z\, \mathrm{d}z\\
& =\int^x \LC u_z\sqrt{1+u_z^2}+\frac{u_zu_{zz}}{\sqrt{1+u_z^2}}\,\partial_z^{-1}u_z\RC \partial_z^{-1}\mathcal{R}^{n-1}u_z\,\mathrm{d}z\\
& =\int^x \LC \sqrt{1+u_z^2}\,\partial_z^{-1}u_z\partial_z^{-1}\mathcal{R}^{n-1}u_z\RC_z \mathrm{d}z=\sqrt{1+u_x^2}\,\partial_x^{-1}u_x\partial_x^{-1}\mathcal{R}^{n-1}u_x.
\end{aligned}
\end{equation}

On the other hand, by the transformation \eqref{reciqu}, the corresponding new function $Q(\tau, y)$ satisfies
\begin{equation*}
\cos Q(\tau, y)=\frac{1}{\sqrt{1+u_x^2(t, x)}}.
\end{equation*}
Differentiating the above expression with respect to $t$ and using the transformation \eqref{recixy}, we have
\begin{equation}\label{thm3.1-3}
\sin Q\cdot \big( Q_{\tau}+Q_y\,y_t\big)=\frac{u_x}{(1+u_x^2)^{\frac{3}{2}}}\,u_{xt},
\end{equation} 
which together with the relationship \eqref{reci-1}, \eqref{reci-2}, \eqref{thm3.1-2} and the equation \eqref{thm3.1-1} gives rise to
\begin{equation*}
\begin{aligned}
\sin Q\cdot & \big( Q_{\tau}+Q_y\sqrt{1+u_x^2}\,\partial_x^{-1}u_x\partial_x^{-1}\mathcal{R}^{n-1}\,u_x \big)\\
&=\frac{u_x}{(1+u_x^2)^{\frac{3}{2}}}\,\partial_x^{-1}\mathcal{R}^{n-1}u_x+\frac{u_x}{(1+u_x^2)^{\frac{3}{2}}}\,\big( u_x^2+u_{xx}\partial_x^{-1}u_x\big)\partial_x^{-1}\mathcal{R}^{n-1}u_x\\
&=\frac{u_x}{\sqrt{1+u_x^2}}\,\partial_x^{-1}\mathcal{R}^{n-1}u_x-\LC \frac{1}{\sqrt{1+u_x^2}}\RC_x\partial_x^{-1}u_x\partial_x^{-1}\mathcal{R}^{n-1}u_x\\
&=\sin Q\cdot \partial_x^{-1}\mathcal{R}^{n-1}\sqrt{1+u_x^2}\sin Q-(\cos Q)_y\sqrt{1+u_x^2}\,\partial_x^{-1}u_x\partial_x^{-1}\mathcal{R}^{n-1}u_x.
\end{aligned}
\end{equation*}
Hence, we derive for $Q(\tau, y)$ the following equation
\begin{equation*}
Q_{\tau}=\partial_x^{-1}\mathcal{R}^{n-1}\sqrt{1+u_x^2}\,\sin Q
\end{equation*}
and then
\begin{equation*}
\partial_x^{-1}\LC \mathcal{R}^{-1}\RC^{n-1}\partial_x\,Q_{\tau}=\partial_x^{-1}\sqrt{1+u_x^2}\,\sin Q.
\end{equation*}
Thanks to Lemma \ref{l3.1}, we deduce by \eqref{reci-1} that $Q(\tau, y)$ satisfies
\begin{equation*}
\bar{\mathcal{R}}^{n-1}\,Q_{\tau}=\partial_y^{-1}\,\sin Q,
\end{equation*}
which is exactly the $(SG)_{-(n-1)}$ equation \eqref{sg-n} for $n=1, 2, \ldots$.

For the remaining case $n=0$, plugging $u_t=K_1[u]=u_x$ into \eqref{thm3.1-2} and \eqref{thm3.1-3} yields
\begin{equation*}
\sin Q\cdot \LC Q_{\tau}+Q_y\int^x \frac{u_zu_{zz}}{\sqrt{1+u_z^2}}\,\mathrm{d}z\RC=\frac{u_xu_{xx}}{(1+u_x^2)^{\frac{3}{2}}}.
\end{equation*}
Combining integration by parts and the sufficiently fast decay property of $u(t, x)$ as $|x|\to +\infty$, using the relationship \eqref{reci-1} and \eqref{reci-2}, we derive from the above identity 
\begin{equation*}
\sin Q\cdot \LC Q_{\tau}+Q_y\int^x \partial_z\sqrt{1+u_z^2}\,\mathrm{d}z\RC=-\LC \frac{1}{\sqrt{1+u_x^2}}\RC_x
\end{equation*}
and then
\begin{equation*}
\sin Q\cdot \LC Q_{\tau}+Q_y\sqrt{1+u_x^2}-Q_y\RC=-(\cos Q)_y\sqrt{1+u_x^2}=\sin Q\cdot Q_y \sqrt{1+u_x^2}.
\end{equation*}
We obtain that $Q(\tau, y)$ satisfies the $(SG)_1$ equation $Q_{\tau}=Q_y$.

Conversely, if $Q(\tau, y)$ is a solution of the (SG)$_{-(n-1)}$ equation for integers $n\geq 0$, since the transformations  \eqref{recixy} and  \eqref{reciqu} are the bijections, tracing the previous steps backwards suffices to verify that the reverse argument is also true. Part (i) is thereby proved.

{\bf(ii).}  Now, we investigate the $(SP)_{-(n-1)}$ equation \eqref{sp-n} for $n\geq 1$. Suppose that $u=u(t, x)$ is the solution of \eqref{sp-n}. The $t$-derivative of the corresponding new variable $y$ defined by \eqref{recixy} satisfies
\begin{equation}\label{thm3.1-4}
\begin{aligned}
y_t & =\int^x \frac{u_zu_{zt}}{\sqrt{1+u_z^2}}\,\mathrm{d}z=\int^x \frac{u_z}{\sqrt{1+u_z^2}}\,\partial_z\,\mathcal{R}^{-(n-1)}\partial_z^2\,\frac{u_z}{\sqrt{1+u_z^2}}\, \mathrm{d}z\\ 
& =\int^x \sin Q\cdot \partial_z^2\LC \partial_z^{-1}\mathcal{R}^{-(n-1)}\partial_z\RC\partial_z\sin Q\,\mathrm{d}z=\int^x \sin Q\cdot \partial_z^2\,\bar{\mathcal{R}}^{n-1}\partial_z\sin Q\,\mathrm{d}z\\
& =\int^y \sin Q\cdot \partial_{\xi}\,\frac{1}{\cos Q}\partial_{\xi}\,\bar{\mathcal{R}}^{n-1}Q_{\xi}\,\mathrm{d}\xi=\tan Q\cdot \partial_y\bar{\mathcal{R}}^{(n-1)}Q_y - \partial_y^{-1}Q_y \partial_y\bar{\mathcal{R}}^{(n-1)}Q_y,
\end{aligned}
\end{equation}
where the identites \eqref{reci-1}, \eqref{reci-2} and integration by parts are used. 

Then, for the corresponding new function $Q(\tau, y)$ that is related with the solution $u(t, x)$ through \eqref{reciqu}, differentiating the relationship
\begin{equation*}
\cos Q(\tau, y)=\frac{1}{\sqrt{1+u_x^2(t, x)}}
\end{equation*}
with respect to $t$ leads to
\begin{equation*}
\sin Q\cdot \LC Q_{\tau}+Q_y\,y_t\RC=\frac{u_x}{\big( 1+u_x^2\big)^{\frac{3}{2}}}\,u_{xt}.
\end{equation*}
Using \eqref{reci-1}, \eqref{reci-2} and \eqref{thm3.1-4}, together with the operator identity \eqref{reop-l3.1}, we have
\begin{equation*}
\begin{aligned}
\sin Q\cdot \big( &Q_{\tau}+\tan Q\cdot Q_y\partial_y\bar{\mathcal{R}}^{n-1}Q_y - Q_y\partial_y^{-1}Q_y \partial_y\bar{\mathcal{R}}^{(n-1)}Q_y\big)\\
& =\frac{u_x}{\big( 1+u_x^2\big)^{\frac{3}{2}}}\,\partial_x\mathcal{R}^{-(n-1)}\partial_x^2\frac{u_x}{\sqrt{1+u_x^2}}=\sin Q\cdot \frac{1}{1+u_x^2}\partial_x^2\LC \partial_x^{-1}\mathcal{R}^{-(n-1)}\partial_x\RC\partial_x\sin Q\\
& =\sin Q\cdot \cos Q\cdot \partial_y\frac{1}{\cos Q}\partial_y\bar{\mathcal{R}}^{n-1}Q_y.
\end{aligned}
\end{equation*}
Hence,
\begin{equation*}
Q_{\tau}+\tan Q\cdot Q_y\partial_y\bar{\mathcal{R}}^{n-1}Q_y - Q_y\partial_y^{-1}Q_y \partial_y\bar{\mathcal{R}}^{n-1}Q_y=\partial_y^2\bar{\mathcal{R}}^{n-1}Q_y+\cos Q \LC\frac{1}{\cos Q}\RC_y\partial_y\bar{\mathcal{R}}^{n-1}Q_y,
\end{equation*}
which implies
\begin{equation*}
Q_{\tau}=\LC \partial_y^2+Q_y\partial_y^{-1}Q_y \partial_y\RC\bar{\mathcal{R}}^{n-1}Q_y=\bar{\mathcal{R}}^nQ_y
\end{equation*}
and verifies that the corresponding $Q(\tau, y)$ satisfies the $(SG)_n$ equation \eqref{sgn+1} for each integer $n\geq 1$. The converse resluts follow from the fact that \eqref{recixy} and \eqref{reciqu} are the bijections. We thus complete the proof of Theorem \ref{t3.1} for all $l\in \mathbb{Z}$.
\end{proof}

\section{The correspondence between the Hamiltonian conservation laws of the short pulse and sine-Gordon equations}

According to the Magri's scheme, one can recursively construct the infinite hierarchy of Hamiltonian conservation laws for the  bi-Hamiltonian integrable systems. In particular, for the SP equation \eqref{sp1}, the corresponding recursive formula
\begin{equation}\label{spclhie}
\mathcal{K}\,\frac{\delta \mathcal{H}_{n-1}}{\delta u}=\mathcal{J}\,\frac{\delta \mathcal{H}_n}{\delta u},\qquad n\in  \mathbb{Z},
\end{equation}
formally provide an infinite collection of the Hamiltonian conservation laws, where $\mathcal{K}$ and $\mathcal{J}$ are the two compatible Hamiltonian operators \eqref{sp-bihaop} admitted by the SP equation. While, for the SG equation \eqref{sg1}, we determine the involved Hamiltonian conservation laws $\bar{\mathcal{H}}_{n}$ by
\begin{equation}\label{sgclhie}
\bar{\mathcal{K}}\, \frac{\delta \bar{\mathcal{H}}_{n-1}}{\delta Q}=\bar{\mathcal{J}}\, \frac{\delta \bar{\mathcal{H}}_n}{\delta Q},\qquad n\in \mathbb{Z},
\end{equation}
using the Hamiltonian pair $\bar{\mathcal{K}}$ and $\bar{\mathcal{J}}$ defined in \eqref{haop-sg} .

In this section, we establish the correspondence between the  two hierarchies of Hamiltonian conservation laws $\{\mathcal{H}_n\}$ and $\{\bar{\mathcal{H}}_n\}$ subject to the transformations \eqref{recixy} and \eqref{reciqu} and prove Theorem \ref{t1.2}. Let us begin with the following two lemmas. 

\begin{lemma}  \label{l4.1}
Let $\{\mathcal{H}_n\}$ and $\{\bar{\mathcal{H}}_n\}$  be the hierarchies of Hamiltonian functionals determined by the recursive formulae \eqref{spclhie} and \eqref{sgclhie}, respectively. Then their corresponding variational derivatives satisfy the following relationship
\begin{equation}\label{vd}
\frac{\delta \mathcal{H}_{-n}}{\delta u}=\partial_y^{-1}\,\frac{\delta \bar{\mathcal{H}}_{n+1}}{\delta Q},\qquad n\in \mathbb{Z},
\end{equation}
under the transformations \eqref{recixy} and \eqref{reciqu}.
\end{lemma}

\begin{proof}
To prove this lemma, we use the induction argument. First of all, we consider the case of  $n\geq 0$. Using \eqref{sg-h0h1} and \eqref{sp-h0}, we have
\begin{equation}\label{l4.1-1}
\frac{\delta \mathcal{H}_{0}}{\delta u}=\frac{u_{xx}}{\big( 1+u_x^2\big)^{\frac{3}{2}}} \quad \mathrm{and} \quad \frac{\delta \bar{\mathcal{H}}_1}{\delta Q}=Q_{yy}.
\end{equation}
In view of the relationship \eqref{reci-1} and \eqref{reci-2}, performing the $x$-derivative for
\begin{equation*}
\cos Q(\tau, y)=\frac{1}{\sqrt{1+u_x^2(t, x)}}
\end{equation*}
leads to 
\begin{equation*}
\sin Q\cdot Q_y\,\sqrt{1+u_x^2}=\frac{u_xu_{xx}}{\big( 1+u_x^2\big)^{\frac{3}{2}}},
\end{equation*}
which together with \eqref{l4.1-1} verifies that \eqref{vd} holds for $n=0$. 

Now, suppose, by induction, that \eqref{vd}  holds for $n=k$ with $k\geq 0$, in other words
\begin{equation*}
\frac{\delta \mathcal{H}_{-k}}{\delta u}=\partial_y^{-1}\,\frac{\delta \bar{ \mathcal{H}}_{k+1}}{\delta Q}.
\end{equation*}
Then, for $n=k+1$, in view of the recursive formulae \eqref{spclhie} and \eqref{sgclhie}, 
\begin{eqnarray*}
\frac{\delta \mathcal{H}_{-(k+1)}}{\delta u}=\mathcal{K}^{-1}\mathcal{J}\frac{\delta \mathcal{H}_{-k}}{\delta u}=\mathcal{K}^{-1}\mathcal{J}\partial_y^{-1}\,\frac{\delta \bar{ \mathcal{H}}_{k+1}}{\delta Q}=\bar{\mathcal{K}}\bar{\mathcal{J}}^{-1}\partial_y^{-1}\bar{\mathcal{K}}^{-1}\bar{\mathcal{J}}\frac{\delta \bar{ \mathcal{H}}_{k+2}}{\delta Q}=\partial_y^{-1}\,\frac{\delta \bar{ \mathcal{H}}_{k+2}}{\delta Q},
\end{eqnarray*}
where we have made use of the identity
\begin{equation}\label{l4.1-2}
\bar{\mathcal{K}}\bar{\mathcal{J}}^{-1}=\mathcal{K}^{-1}\mathcal{J},
\end{equation}
which arises from \eqref{reop-l3.1} with $n=1$. This establishes the induction step and thus proves \eqref{vd} for each integer $n\geq 0$.

Next, we deal with the case for the integers $n\leq -1$. From \eqref{sg-h0h1}, \eqref{sp-ha12} and \eqref{reci-2}, we derive
\begin{equation*}
\frac{\delta \bar{\mathcal{H}}_0}{\delta Q}=\sin Q \quad \mathrm{and} \quad \frac{\delta \mathcal{H}_1}{\delta u}=u=\partial_y^{-1}\sin Q.
\end{equation*}
So \eqref{vd} holds for $n=-1$. Assume now that \eqref{vd}  holds for $n=k$ with $k\leq -1$. Then, for $n=k-1$, thanks to \eqref{spclhie} and \eqref{sgclhie} and using \eqref{l4.1-2} again, we arrive at
\begin{eqnarray*}
\frac{\delta \mathcal{H}_{-(k-1)}}{\delta u}=\mathcal{J}^{-1}\mathcal{K}\frac{\delta \mathcal{H}_{-k}}{\delta u}=\bar{\mathcal{J}}\bar{\mathcal{K}}^{-1}\partial_y^{-1}\frac{\delta \bar{ \mathcal{H}}_{k+1}}{\delta Q}=\bar{\mathcal{J}}\bar{\mathcal{K}}^{-1}\partial_y^{-1}\bar{\mathcal{J}}^{-1}\bar{\mathcal{K}}\frac{\delta \bar{ \mathcal{H}}_{k}}{\delta Q}=\partial_y^{-1}\frac{\delta \bar{ \mathcal{H}}_{k}}{\delta Q}.
\end{eqnarray*}
Therefore, a straightforward  induction verifies \eqref{vd} for $n\leq -1$. This completes the proof of the lemma in general.
\end {proof}

The preceding lemma reveals the one-to-one correspondence between the variational derivatives of the corresponding Hamiltonian conservation laws admitted by  the SP and SG equations.  In order to investigate the effect of  the transformations \eqref{recixy} and \eqref{reciqu} on the two hierarchies of the Hamiltonian conservation laws $\{\mathcal{H}_n\}$ and $\{\bar{\mathcal{H}}_n\}$,  we require a formula for the change of the variational derivatives with respect to $u$ and $Q$, respectively.

\begin{lemma} \label{l4.2}
Let $u(t, x)$ and $Q(\tau, y)$ be related by the transformations \eqref{recixy} and \eqref{reciqu}. Assume further that $u(t, x)$ satisfies the equation $u_t=K_{l+1}[u]$ in the SP hierarchy and $Q(\tau, y)$ is the solution of the $(SG)_{-(l-1)}$ equation $Q_{\tau}=\bar{K}_{-(l-1)}[Q]$ with some $l\in \mathbb{Z}$ in terms of Theorem \ref{t3.1}. If $\bar{\mathcal{H}}_n(Q)$ is one Hamiltonian functional for $Q_{\tau}=\bar{K}_{-(l-1)}[Q]$ and $\bar{\mathcal{H}}_n(Q)=\mathcal{H}(u)$ under transformations \eqref{recixy} and \eqref{reciqu}, then
\begin{equation}\label{vdqu}
\frac{\delta \mathcal{H}(u)}{\delta u}=-\bar{\mathcal{K}}\frac{\delta \bar{\mathcal{H}}_n(Q)}{\delta Q},
\end{equation}
$\bar{\mathcal{K}}$ is the Hamiltonian operator given by \eqref{haop-sg}.
\end{lemma}

\begin{proof}
First of all, in view of \eqref{reciqu}, we denote 
\begin{equation*}
Q(\tau, y)=F[u(t, x)]\equiv \arccos\frac{1}{\sqrt{1+u_x^2(t,x)}}.
\end{equation*}
According to \eqref{recixy}, \eqref{reci-1} and \eqref{reci-2}, we derive with a test function $\rho\in \mathcal{C}^{\infty}_c$ that
\begin{eqnarray*}
\begin{aligned}
\frac{\mathrm{d}}{\mathrm{d}\epsilon}\Big|_{\epsilon=0}y(u+\epsilon \rho) &=\int^x \frac{u_z \rho_z}{\sqrt{1+u_z^2}}\,\mathrm{d}z\\
 &=\frac{u_x}{\sqrt{1+u_x^2}}\rho - \int^x \partial_z\LC\frac{u_z}{\sqrt{1+u_z^2}}\RC \rho\,\mathrm{d}z=\sin Q\cdot \rho - \int^y \cos Q\cdot Q_{\xi}\,\rho\,\mathrm{d}\xi.
\end{aligned}
\end{eqnarray*}
Hence, on the one hand,
\begin{eqnarray*}
\begin{aligned}
\frac{\mathrm{d}}{\mathrm{d}\epsilon}\Big|_{\epsilon=0}F[u+\epsilon \rho] &=Q_y\frac{\mathrm{d}}{\mathrm{d}\epsilon}\Big|_{\epsilon=0}y(u+\epsilon \rho)+\frac{\mathrm{d}}{\mathrm{d}\epsilon}\Big|_{\epsilon=0}^{y\, fixed}F[u+\epsilon \rho]\\
&=Q_y\LC \sin Q\cdot \rho - \partial_y^{-1}\cos Q\,Q_y\,\rho \RC+\frac{\mathrm{d}}{\mathrm{d}\epsilon}\Big|_{\epsilon=0}^{y\, fixed}F[u+\epsilon \rho].
\end{aligned}
\end{eqnarray*}
On the other hand, in terms of the Fr\'echet derivative
\begin{equation*}
\frac{\mathrm{d}}{\mathrm{d}\epsilon}\Big|_{\epsilon=0}F[u+\epsilon \rho]=\mathcal{D}_{F[u]}(\rho)=\frac{1}{1+u_x^2}\partial_x \rho=\cos Q\cdot \partial_y \rho.
\end{equation*}
We thus deduce that 
\begin{equation}\label{lem3.2-1}
\frac{\mathrm{d}}{\mathrm{d}\epsilon}\Big|_{\epsilon=0}^{y\, fixed}F[u+\epsilon \rho]=\big( \cos Q \cdot \partial_y-Q_y\sin Q+Q_y\partial_y^{-1}Q_y\sin Q\big)\rho.
\end{equation}

Next, by the assumption, we have
\begin{eqnarray*}
\frac{\mathrm{d}}{\mathrm{d}\epsilon}\Big|_{\epsilon=0} \mathcal{H}(u+\epsilon \rho)=\frac{\mathrm{d}}{\mathrm{d}\epsilon}\Big|_{\epsilon=0} \bar{\mathcal{H}}_n\left(F[u+\epsilon \rho]\right).
\end{eqnarray*}
Then, the usual definition of variational derivative leads to 
\begin{equation}\label{lem3.2-2}
\frac{\mathrm{d}}{\mathrm{d}\epsilon}\Big|_{\epsilon=0}\mathcal{H}(u+\epsilon \rho)=\int \frac{\delta \mathcal{H}}{\delta u}\,\cdot \rho\,\mathrm{d}x.
\end{equation}
While, in view of  \eqref{lem3.2-1} and integration by parts, we obtain 
\begin{eqnarray*}
\begin{aligned}\label{lem3.2-3}
\frac{\mathrm{d}}{\mathrm{d}\epsilon} & \Big|_{\epsilon=0} \bar{\mathcal{H}}_n\left(F[u+\epsilon \rho]\right) =\int \frac{\delta \bar{\mathcal{H}}_n}{\delta Q}\,\cdot \frac{\mathrm{d}}{\mathrm{d}\epsilon} \Big|_{\epsilon=0}^{y\, fixed}F[u+\epsilon \rho]\,\mathrm{d}y\\
&=\int \frac{\delta \bar{\mathcal{H}}_n}{\delta Q}\,\cdot \big( \cos Q \cdot \partial_y-Q_y\sin Q+Q_y\,\partial_y^{-1}Q_y\cos Q\big)\,\rho\,\mathrm{d}y\\
&=-\int \rho \big( \partial_y \cos Q+Q_y \sin Q+\cos Q\cdot Q_y\,\partial_y^{-1}\,Q_y\big) \frac{\delta \bar{\mathcal{H}}_n}{\delta Q}\,\mathrm{d}y+\int \cos Q\cdot Q_y\,\rho\,\mathrm{d}y \int \frac{\delta \bar{\mathcal{H}}_n}{\delta Q}\,Q_y\,\mathrm{d}y\\
&=-\int \rho \big( \cos Q \partial_y+\cos Q\cdot Q_y\,\partial_y^{-1}\,Q_y\big) \frac{\delta \bar{\mathcal{H}}_n}{\delta Q}\,\mathrm{d}y\\
&=-\int \rho \big( \partial_y+Q_y\,\partial_y^{-1}\,Q_y\big) \frac{\delta \bar{\mathcal{H}}_n}{\delta Q}\,\mathrm{d}x=-\int \rho \, \bar{\mathcal{K}} \frac{\delta \bar{\mathcal{H}}_n}{\delta Q}\,\mathrm{d}x.
\end{aligned}
\end{eqnarray*}
Here, we delete the boundary value in the integration by parts. Indeed, since $\bar{\mathcal{H}}_1(Q)=-\int Q_y^2/2\,dy$ is one of the Hamiltonian functionals for the equation $Q_{\tau}=\bar{K}_{-(l-1)}[Q]$ which belongs to the SG hierarchy \eqref{hie-sg}, we then deduce from the convolution property for the hierarchy of the Hamiltonian functionals that
\begin{equation*}
\int \frac{\delta \bar{\mathcal{H}}_n}{\delta Q}\,Q_y\,\mathrm{d}y=\int \frac{\delta \bar{\mathcal{H}}_n}{\delta Q}\,\partial_y^{-1}Q_{yy}\,\mathrm{d}y=\int \frac{\delta \bar{\mathcal{H}}_n}{\delta Q}\,\bar{\mathcal{J}}\frac{\delta \bar{\mathcal{H}}_1}{\delta Q}\,\mathrm{d}y=0,
\end{equation*}
where $\bar{\mathcal{J}}=\partial_y^{-1}$ is one of the Hamiltonian operators admitted by the SG hierarchy. Therefore, we conclude from \eqref{lem3.2-2} that \eqref{vdqu} holds true and prove the lemma.
\end{proof}

Finally, under the hypothesis of Lemma \ref{l4.2}, we define the functional
\begin{equation*}
\mathcal{H}(u)\equiv \bar{\mathcal{H}}_n(Q)
\end{equation*}
for each $n\in \mathbb{Z}$. From Lemma \ref{l4.2}, we see that  
\begin{equation*}
\frac{\delta \mathcal{H}(u)}{\delta u}=-\bar{\mathcal{K}}\frac{\delta \bar{\mathcal{H}}_n(Q)}{\delta Q}.
\end{equation*}
Meanwhile, Lemma \ref{l4.1} and the recrsive formula \eqref{spclhie} together with the relationship \eqref{l4.1-2} readily lead to 
\begin{equation*}
\frac{\delta \mathcal{H}(u)}{\delta u}=-\bar{\mathcal{K}}\bar{\mathcal{J}}^{-1}\frac{\delta \mathcal{H}_{-(n-1)}(u)}{\delta u}=-\mathcal{K}^{-1}\mathcal{J}\frac{\delta \mathcal{H}_{-(n-1)}(u)}{\delta u}=-\frac{\delta \mathcal{H}_{-n}(u)}{\delta u}.
\end{equation*}
The above identity allows us safely draw the conclusion that
 \begin{equation*}
\mathcal{H}(u)=-\mathcal{H}_{-n}(u),
\end{equation*}
from which
\begin{equation*}
 \mathcal{H}_{-n}(u)=-\bar{ \mathcal{H}}_{n}(Q), \qquad n\in \mathbb{R},
\end{equation*}
follows. Assembling the previous two lemmas, we conclude that, there exists an one-to-one correspondence between the two sequences of the Hamiltonian conservation laws $\{\mathcal{H}_n\}$ and $\{\bar{\mathcal{H}}_n\}$. We thus have proved 
Theorem \ref{t1.2}.

\vglue .83cm

\noindent {\bf Acknowledgements.}
The work of J. Kang is supported by NSF-China grant 11471260 and the Foundation of Shannxi Education Committee-12JK0850. 
The work of X.C. Liu is supported by NSF-China grant 11401471 and Ph.D. Programs Foundation of Ministry of Education of China-20136101120017. 

\vglue .35cm

\vglue .5cm


\begin{thebibliography}{99}

\vglue .85cm


\bibitem{as}{\small\textsc{M.J. Ablowitz and H. Segur,}\ Solitons and the Inverse Scattering Transform, SIAM, Philadelphia, PA, 1981.}


\bibitem{brt} {\small \textsc{R. Beals, M.L. Rabelo, and K. Tenenblat,}\ B\"{a}cklund transformations and inverse scattering solutions for some pseudospherical surface equations, {\it Stud. Appl. Math.}, {\bf 81} (1989), 123-151.}


\bibitem{bss98} {\small \textsc{R. Beals, D.H. Sattinger, and J. Szmigielski,}\ Acoustic scattering and the extended Korteweg de Vries hierarchy, {\it Adv. Math.}, {\bf 140} (1998), 190-206.}


\bibitem{bss} {\small \textsc{R. Beals, D.H. Sattinger, and J. Szmigielski,}\ Multipeakons and the classical moment problem, {\it Adv. Math.}, {\bf 154} (2000), 229-257.}


\bibitem{bgr} {\small \textsc{P.M. Bies, P. Gorka, and E. Reyes,}\ The dual modified Korteweg-de Vries-Fokas-Qiao equation: geometry and local analysis, {\it J. Math. Phys.}, {\bf 53} (2012), 073710.}


\bibitem{bru1} {\small \textsc{J.C. Brunelli,}\ The short pulse hierarchy, {\it J. Math. Phys.}, {\bf 46} (2005), 123507.}


\bibitem{bru2} {\small \textsc{J.C. Brunelli,}\ The bi-Hamiltonian structure of the short pulse equation, {\it Phys. Lett. A}, {\bf 353} (2006), 475-478.}










\bibitem{cjsw}{\small \textsc{Y. Chung, C.K.R.T. Jones, T. Sch\"{a}fer, and C.E. Wayne,}\ Ultra-short pulses in linear and nonlinear media, {\it Nonlinearity}, {\bf 18} (2005), 1351-1374.}






































\bibitem{fuc} {\small \textsc{B. Fuchssteiner,}\ Some tricks from the symmetry-toolbox for nonlinear equations: generalizations of the Camassa-Holm equation, {\it Physica D}, {\bf 95} (1996), 229-243.}






\bibitem{gloq}{\small \textsc{G.L. Gui, Y. Liu, P.J. Olver, and C.Z. Qu,}\ Wave-breaking and peakons for a modified Camassa-Holm equation, {\it Commun. Math. Phys.}, {\bf 319} (2013), 731-759.}








\bibitem{kloq}{\small \textsc{J. Kang, X.C. Liu, P.J. Olver, and C.Z. Qu,}\ Liouville correspondence between the modified KdV hierarchy and its dual integrable hierarchy, {\it J. Nonlinear Sci.},  {\bf 26} (2016), 141-170.}




\bibitem{kbk}{\small \textsc{V.K. Kuetche, T.B. Bouetou, and T.C. Kofane,}\ On exact solutions of the Sch\"{a}fer-Wayne short pulse equation: WKI eigenvalue problem, {\it J. Phys. A: Math. Theor.},  {\bf 40} (2007), 5585-5596.}


\bibitem{len0} {\small \textsc{J. Lenells,}\ The correspondence between KdV and Camassa-Holm, {\it Int. Math. Res. Not.}, {\bf 71} (2004), 3797-3811.}










\bibitem{llq} {\small \textsc{X.C. Liu, Y. Liu, and C.Z. Qu,}\ Orbital stability of the train of peakons for an integrable modified Camassa-Holm equation, {\it Adv. Math.}, {\bf 255} (2014), 1-37.}




\bibitem{lps} {\small \textsc{Y. Liu, D. Pelinovsky, and A. Sakovich,}\ Wave breaking in the Short-Pulse equations, {\it Dynm. Partial Differential Equations}, {\bf 6} (2009), 291-310.}


\bibitem{mag}{\small\textsc{F. Magri,}\ A simple model of the integrable Hamiltonian equation, {\it J. Math. Phys.}, {\bf 19} (1978), 1156-1162.}


\bibitem{mats1}{\small\textsc{Y. Matsuno,}\ Multiloop soliton and multibreather solutions of the short pulse model equation, {\it J. Phys. Soc. Japan}, {\bf 76} (2007), 084003.}


\bibitem{mats2}{\small\textsc{Y. Matsuno,}\ Periodic solutions of the short pulse model equation, {\it J. Math. Phys.}, {\bf 49} (2008), 073508.}


\bibitem{mc} {\small \textsc{H.P. McKean,} The Liouville correspondence between the Korteweg-de Vries and the Camassa-Holm hierarchies, {\it Commun. Pure Appl. Math.}, {\bf 56} (2003), 998-1015.}


\bibitem{mil} {\small \textsc{R. Milson,} Liouville transformation and exactly solvable Schr\"{o}dinger equations, {\it Int. J. Theo. Phys.}, {\bf 37} (1998), 1735-1752.}




\bibitem{folv}{\small\textsc{F.W.J. Olver,}\ Asymptotics and Special Functions, Academic Press, New York, 1974.}


\bibitem{olv0} {\small \textsc{P.J. Olver,}\ Evolution equations possessing infinitely many symmetries, {\it J. Math. Phys.}, {\bf 18} (1977), 1212-1215.}


\bibitem{olv}{\small\textsc{P.J. Olver,}\ Applications of Lie Groups to Differential Equations. Second edition, {\it Graduate Texts in Mathematics}, {\bf 107}, Springer-Verlag, New York, 1993.}




\bibitem{or} {\small \textsc{P.J. Olver and P. Rosenau,}\ Tri-Hamiltonian duality between solitons and solitary-wave
 solutions having compact support, {\it Phys. Rev. E}, {\bf 53} (1996), 1900-1906.}




\bibitem{park} {\small \textsc{E.J. Parkes,}\ Some periodic and solitary travelling-wave solutions of the short-pulse equation, {\it Chaos, Solitons and Fractals}, {\bf 38} (2008), 154-159.}


\bibitem{ps} {\small \textsc{D. Pelinovsky and A. Sakovich,}\ Global well-posedness of the Short-Pulse and sine-Gordon equations in energy space, {\it Commun. Partial Differential Equations}, {\bf 35} (2010), 613-629.}






\bibitem{rab} {\small \textsc{M.L. Rabelo,}\ On equations which describe pseudospherical surfaces, {\it Stud. Appl. Math.}, {\bf 81} (1989), 221-248.}












\bibitem{ss1} {\small \textsc{A. Sakovich and S. Sakovich,}\ The short pulse equation is integrable, {\it J. Phys. Soc. Japan}, {\bf 74} (2005), 239-241.}


\bibitem{ss2} {\small \textsc{A. Sakovich and S. Sakovich,}\ Solitary wave solutions of the short pulse equation, {\it J. Phys. A: Math. Gen.}, {\bf 39} (2006), L361-L367.}


\bibitem{ss3} {\small \textsc{A. Sakovich and S. Sakovich,}\ On transformations of the Rabelo equations, {\it Symmetry, Integr. Geom.: Methods Appl.}, {\bf 3} (2007), 086.}


\bibitem{sw} {\small \textsc{T. Sch\"{a}fer and C. E. Wayne,}\ Propagation of ultra-short optical pulses in cubic nonlinear media, {\it Physica D},
{\bf 196} (2004), 90-105.}








\bibitem{wki} {\small \textsc{M. Wadati, K. Konno, and Y.H. Ichikawa,}\ New integrable nonlinear evolution equations, {\it J. Phys. Soc. Japan}, {\bf 47} (1979), 1698-1700.}



\end{thebibliography}
\end{document}